\documentclass[twoside,leqno,twocolumn]{article}

\usepackage[letterpaper]{geometry}
\usepackage{ltexpprt}

\usepackage{graphicx} 
\graphicspath{{./figs/}}
\usepackage{epsfig} 

\graphicspath{ {./figs/} }
\usepackage{times} 
\usepackage{amsmath} 
\usepackage{amssymb}  
\usepackage{algorithmic}
\usepackage{algorithm}
\usepackage[columnwise,switch]{lineno} 

\newcommand{\setof}[1]{\left\{ {#1} \right\}}
\newcommand{\always}{\square}
\newcommand{\eventually}{\lozenge}

\newcommand{\conj}{\wedge}
\newcommand{\disj}{\vee}
\newcommand{\bigconj}{\bigwedge}

\newcommand{\nbhd}[2]{\mathcal{N}_{#1}(#2)}
\newcommand{\enbhd}[1]{\mathcal{N}_{#1}}
\renewcommand{\vec}[1]{\mathbf{#1}}

\newcommand{\graph}{\mathcal{G}}
\newcommand{\edges}{\mathcal{E}}
\newcommand{\nodes}{\mathcal{V}}
\newcommand{\atomic}{\mathcal{AP}}

\newcommand{\buchi}{B{\"u}chi~}
\newcommand{\card}[1]{\lvert #1 \rvert}
\newcommand{\powerset}[1]{2^{#1}}
\newcommand{\sensors}[1]{\mathbf{#1}}

\newtheorem{thm}{Theorem}

\newtheorem{prop}{Proposition}

\newtheorem{prob}{Problem}

\newtheorem{defn}{Definition}

\newtheorem{cor}{Corollary}

\usepackage{xcolor}

\title{A Temporal Logic-Based Hierarchical Network Connectivity Controller
}

\author{Hans Riess\thanks{University of Pennsylvania, Dept.\ of Electrical/Systems Engineering}
\and{Yiannis Kantaros}\thanks{U.~Penn., Dept.\ of Electrical/Systems Eng., Dept.\ of Computer and Information Science}
\and{\stepcounter{footnote}\stepcounter{footnote}\stepcounter{footnote}\stepcounter{footnote} George Pappas}\thanks{U.~Penn., Dept.\ of Electrical/Systems Engineering}
\and{Robert Ghrist}\thanks{U.~Penn., Dept.\ of Electrical/Systems Eng., Dept.\ of Mathematics}}
\begin{document}

\date{}
\maketitle

	
\fancyhead[L]{Preprint: to appear in \textit{2021 Proceedings of the Conference on Control and Its Applications}, SIAM}




\begin{abstract} \small\baselineskip=9pt
In this paper, we consider networks of static sensors with integrated sensing and communication capabilities. The goal of the sensors is to propagate their collected information to every other agent in the network and possibly a human operator. Such a task requires constant communication among all agents which may result in collisions and congestion in wireless communication. To mitigate this issue, we impose locally non-interfering connectivity constraints that must be respected by every agent. We show that these constraints along with the requirement of propagating information in the network can be captured by a Linear Temporal Logic (LTL) framework. Existing temporal logic control synthesis algorithms can be used to design correct-by-construction communication schedules that satisfy the considered LTL formula. Nevertheless, such approaches are centralized and scale poorly with the size of the network. We propose a hierarchical LTL-based algorithm that designs communication schedules that determine which agents should communicate while maximizing network usage. We show that the proposed algorithm is complete and demonstrate its efficiency and scalability through analysis and numerical experiments.
\end{abstract}

\section{Introduction}
\label{sec:introduction}
The 5G standard \cite{agiwal2016} offers drastic improvements in latency, base-station capacity, and data rates, and moves away from the base-station centered network topology in favor of a more complex device-centered topology.
In low-range wireless and LIDAR (Light Detection and Ranging) networks, especially sensor networks, the topology of both the network and the coverage region can be uncertain or unknown to command units.
These considerations highlight the need for more sophisticated scheduling protocols informed by the topology of the network.
One such approach to optimize network usage while avoiding collisions is to allow the presence of simultaneous access to the network while satisfying a local non-interference rule.

In this paper, we consider networks of static sensors with integrated sensing and communication capabilities.
The sensors collect information about the ambient environment which needs to be available to everyone in the network.
Such tasks require constant communication among all agents which may result in collisions and  congestion in wireless communication. To mitigate this issue, we propose a novel hierarchical algorithm that designs communication schedules that determine which agents should communicate (i.e., which communication links should be activated) and when communication should happen while optimizing network usage. The schedules are designed by solving an optimal temporal logic control problem. In particular, we define a Linear Temporal Logic formula \cite{baier2008} that requires all available communication links among agents to be activated infinitely often (liveness) to ensure that information can be propagated across the network while respecting local non-interference constraints that require each agent to connect to only one agent at a time. Schedules that satisfy the considered temporal logic specification  can be computed by employing existing temporal logic control synthesis algorithms \cite{smith2011optimal,kloetzer2009automatic,fainekos2005hybrid,baier2008} or off-the-shelf model checkers \cite{holzmann2004spin}. Nevertheless, such approaches are centralized and scale poorly with the size of the network.  To mitigate this issue, we propose a hierarchical approach that relies on decomposing the network into smaller networks via commanding agents that compute plans for sensors in their range. 

We show through simulation studies that the proposed \emph{Locally Non-interfering Connectivity} (LNC) method can be used to design schedules for large static sensor networks. Our key contribution to the literature is a distributed method for satisfying linear temporal logic constraints; the particular focus of this work is synthesizing a link activation schedule subject to non-interfering connectivity constraints, but our methods translate, without any modification, to any network multi-agent system with local temporal logic constraints. As a secondary contribution, we open a new line of research that studies information propagation over graphs subject to semantic constraints (e.g. in linear temporal logic or other modal logics).

\subsection*{Related Research}
Several works have been proposed to design controllers that ensure point-to-point or end-to-end network connectivity of \textit{mobile} robot networks for all time.
Such controllers either rely on graph theoretic approaches  \cite{zavlanos2007potential,ji2007distributed,zavlanos2008distributed,sabattini2013decentralized,zavlanos2011graph} or employ more realistic communication models that take into account path loss, shadowing, and multi-path fading as well as optimal routing decisions for desired information rates
\cite{zavlanos2013network,yan2012robotic,kantaros2016distributed,stephan2017concurrent}. Intermittent connectivity methods that allow the mobile agents to temporarily get disconnected to accomplish their tasks and occasionally revert to connected configurations have also been proposed \cite{hollinger2010multi,kantaros2016}.
Connectivity scheduling problems for static networks, similar to the one considered here, are addressed in \cite{zavlanos2008,chatzipanagiotis2016}. In particular, \cite{zavlanos2008} addresses the control of switching networks via Laplacian dynamics in which nodes can alternate between \textit{sleep} and \textit{active} states while guaranteeing a multi-hop path to and from a subset of boundary nodes. Note that \cite{zavlanos2008} addresses a point-to-point connectivity problem while here 
our goal is to ensure that information collected by any sensor at any time will eventually be propagated to any other sensor. A connectivity scheduling problem is posed and solved via mixed-integer programming in \cite{chatzipanagiotis2016}, although the goals of the problem addressed there (i.e. providing communication services to static nodes) are different from ours. As a novel contribution of this paper, we furnish a distributed algorithm to design communication schedules over wireless sensor networks that maximize network usage and respect non-interference constraints as specified by an LTL formula.
\subsection*{Outline}
The rest of the paper is summarized as follows. In Section \ref{sec:problem}, we formally define the scheduling connectivity problem with non-interference constraints that is solved by a distributed algorithm presented in Section \ref{sec:solution}. In Section \ref{sec:simulations}, we present simulation studies.
In Section \ref{sec:discussion}, we discuss a blueprint for a completely decentralized solution to our scheduling problem, as well as posit a modification of our main Algorithm to better suit networks whose links are not determined by proximity.

\section{Problem Formulation}
\label{sec:problem}
Consider a collection of $N$ stationary agents located at positions $x_i$, $i\in\{1,\dots,N\}$ with integrated sensing and wireless communication capabilities that are placed in key locations where continuous sensing and communication is required.
Assume $\sensors{X} = \setof{x_i}_{i=1}^N$ is a subset of $
\sensors{Z}$, a compact subset of $\mathbb{R}^d$.
Each agent collects information that needs to be propagated to all other agents.
Specifically, we assume that each agent is capable of communicating with any other agent that is within range $r>0$.
This setup can be modeled as an undirected graph $\graph = \left( \nodes, \edges \right)$ where $\nodes=\{1,\dots,N\}$ is the set of nodes and $\edges=\setof{ij|~\lVert x_i-x_j\rVert\leq r} \subseteq \nodes\times\nodes$ is a set of edges, or links.
For any node $i\in\nodes$, let $\mathcal{N}_i$ be a set that collects all neighbors of $i$, i.e.,
$\mathcal{N}_i=\{j\neq i|~ij\in\edges\}$. Also, for any edge $ij$ we define the \emph{neighborhood of a link}, $\enbhd{ij} = \{i'j'| \left( i'= i, j' \in \mathcal{N}_i \right) \vee \left( j'=j, i' \in \mathcal{N}_j \right)\}$ that collects all edges associated with nodes $i$ and $j$. For convenience, let $B_r(x) = \setof{y~\vert~ \|x - y\| < r}$ denote the open ball of radius $r$ centered around $x$.

To ensure the propagation of information, frequent intermittent communication among all agents is required.
This may result in communication congestion, especially in large and dense networks.
Therefore, we impose locally non-interfering communication constraints.
In particular, we require that when communication between nodes $i$ and $j$ happens, i.e., when the communication link $ij$ is active, all other communication links in the set $\enbhd{ij}$ are deactivated.
We say that the network is \emph{locally non-interfering} at link $e \in \edges$, when all links $e' \in \enbhd{ij}$ are deactivated.
Also, the network is \emph{globally non-interfering} if it is locally non-interfering for every edge $e \in \edges$.
Our goal is to compute a schedule that determines the order in which the links should be activated so that (i) all links are activated infinitely often to ensure that information collected by any sensor node will eventually be transmitted to all other nodes while (ii)  ensuring globally non-interfering communication all the time.

To formally model requirements (i)-(ii), we employ Linear Temporal Logic (LTL). The basic ingredients of LTL are a set of atomic propositions $\mathcal{AP}$, the boolean operators, i.e., conjunction $\wedge$, and negation $\neg$, and two temporal operators, next $\bigcirc$ and until $\mathcal{U}$. LTL formulas over a set $\mathcal{AP}$ can be constructed based on the following grammar: $\phi::=\texttt{true}~|~\pi~|~\phi_1\wedge\phi_2~|~\neg\phi~|~\bigcirc\phi~|~\phi_1~\mathcal{U}~\phi_2$, where $\pi\in\mathcal{AP}$. For brevity, we abstain from presenting the derivations of other Boolean and temporal operators, e.g., \textit{always} $\always$, \textit{eventually} $\eventually$, \textit{infinitely often} $\always \eventually$, \textit{at some point forever} $\eventually \always$, \textit{disjunction} $\disj$, and \textit{implication} $\Rightarrow$, which can be found in \cite{baier2008}. Specifically, requirements (i)-(ii) can be captured by the following LTL specification:
\begin{equation}
\label{eq:global-phi}
\varphi = \bigconj_{i j \in \mathcal{E}} \varphi_{i j},
\end{equation}
that requires the LTL sub-formula $\varphi_{i j}$ to be true for all links $ij\in\mathcal{E}$, where
\begin{equation}
\label{eq:local-phi}
\varphi_{i j} = \underbrace{\always \left(  \pi^{i j} \Rightarrow \bigconj_{i' j' \in \mathcal{N}_{i j} \setminus \{ ij\}} \neg \pi^{i' j'} \right)}_{\text{(a)}} \conj \,  \underbrace{\always \eventually \pi^{i j}}_{\text{(b)}},
\end{equation}
where $\pi^{i j}$ is a Boolean variable that is true whenever link $i j$ is activated. Note that  formula \eqref{eq:local-phi} requires link $ij$ to be activated infinitely often while respecting the locally non-interference constraints, i.e., that all other links $e'\in\enbhd{ij} $ are deactivated, as captured by part (b) and (a), respectively.  

\subsection*{Modeling Link Activations as Transition Systems}

To design a schedule for all agents, i.e., infinite sequences of activated links, that satisfy \eqref{eq:global-phi}, we view the discrete-time dynamical system of link activations as a transition system, $\vec{TS}$.
First, we model a single link $ij$ as a transition system. 

\begin{defn}
Given a link $ij \in \edges$, we define the \emph{link transition system}, $TS^{i j} = ( Q^{i j}, q^{i j}_{0}, \Sigma^{i j}, \rightarrow_{i j}, \atomic^{i j}, o^{ij})$.
\begin{itemize}
\item $Q^{i j} = \setof{q_0^{ij}, q_1^{ij}}$ is the set of states where the states $q_0^{ij}$ and $q_1^{ij}$ encode that the link $ij$ is activated and deactivated, respectively.
\item $q^{i j}_0$ is an initial state (i.e., initially $ij$ is deactivated).
\item $\Sigma^{ i j} = \setof{\sigma_0^{ij}, \sigma_1^{ij}}$ is the control input alphabet, \textit{standby} or \textit{switch}.
\item $q_{k} \xrightarrow{\sigma_l^{i j}}_{ij} q_{k'}$, $l=0,1$ are state transitions,
\begin{itemize}
\item $q_{k} \xrightarrow{\sigma_0^{ij}}_{i j} q_{k}$ for $k = 0, 1$;
\item $q_{k} \xrightarrow{\sigma_1^{ij}}_{i j} q_{k'}$ for $k = 0, k' = 1$ or $k = 1, k' = 0$.
\end{itemize}
\item $\atomic^{i j} = \setof{\pi^{i j}}$ are the atomic propositions.
\item $o^{ij}(q_0) = \pi^{i j}$, $o(q_1) = \neg \pi^{i j}$ are observations of states.
\end{itemize}
\end{defn}

Then, we define a transition system modeling the entire network.
\begin{defn}
From a collection of $TS^{ i j}$, define the \textit{product transition system},
$\vec{TS} = (Q, Q_0, \Sigma, \rightarrow, \atomic, o)$.
\begin{itemize}
\item $Q = \setof{q^E}_{E \in \powerset{\edges}}$ are states.
\item $Q_0 \subseteq Q$ are initial states.
\item $\Sigma = \prod_{i j \in \mathcal{E}} \Sigma^{ij}$.
\item $\rightarrow \subseteq Q \times \Sigma \times Q$: there is a transition between every pair of states given by exactly one control input $\sigma \in \Sigma$. (Alternative transition systems could be designed, e.g.~by subsampling $\rightarrow$.)
\item $\atomic = \setof{\pi^{i j}}_{i j \in \edges}$ are the atomic propositions.
\item $o: Q \rightarrow \powerset{\atomic}$ given by $q^E \rightsquigarrow \setof{\pi^{i j}}_{i j \in E}$ is the observation map.
\end{itemize}
\end{defn}

We enrich our transition system with an associated cost to transition from one state to another. In task-coordination problems in robotics, traveling time between regions a robot is to visit is a reasonable choice of cost \cite{guo2015}. In our case, a reasonable choice of cost is a suitable metric on $\powerset{\edges}$ (which indexes $Q$). Possibilities include the Jaccard distance (e.g.~in \cite{onella2007}), and the Hausdorff distance between collections of activated sensor coordinates (e.g.~in \cite{zhu2009}).
Transition costs can encourage goals manifold, including but not limited to (i) fostering propagation of information, (ii) reducing energy costs associated to link rerouting, connection activation, or transmission power consumption, and (iii) avoiding trivial or undesirable (inefficient) traces (i.e. activate one edge at a time). The Jaccard distance used in our experiments (Section \ref{sec:simulations}),
\begin{equation}
	d_{\mathrm{jac}}(A,B) = 1 - \frac{|A \cap B |}{| A \cup B |},
\end{equation} penalizes switching control inputs, consequently, avoiding the trivial solution.

\begin{defn}
A \emph{weighted transition system} with \textit{cost function}, $c: \Sigma \rightarrow \mathbb{R}_{+}$, is a pair, $(\vec{TS}, c)$, where
\begin{equation}
c: \left( q^{E} \xrightarrow{\sigma \in \Sigma} q^{E'} \right) \rightsquigarrow c(\sigma).
\end{equation}
\end{defn}
For a \textit{finite} trace $\tau \in Q^\omega$, define the \textit{cost to go},
\begin{equation}
\label{eq:cost-to-go}
\hat{J}(\tau) = \sum_{t=1}^{|\tau|-1} c \left( \tau({t-1}) \xrightarrow{\sigma_t} \tau(t) \right).
\end{equation}

Our goal is to compute a schedule $\tau=\tau(0)\tau(1)\tau(2)\dots$ defined as an infinite sequence of states of $\vec{TS}$ that satisfies $\varphi$. In other words, $\tau$ determines the order in which links are activated and $\tau(k)\in Q$ collects the links that should be active at the $k$-th discrete time instant. Given any LTL
formula $\varphi$, if there exists a schedule $\tau$ satisfying $\varphi$, then it can be written in a finite representation, called prefix-suffix structure, i.e., $\tau=\tau_{\text{pre}}[\tau_{\text{suf}}]^{\omega}$, where the  prefix part $\tau_{\text{pre}}=\tau(0),\tau(1),\dots,\tau(K)$ is executed only once followed by the indefinite execution of the suffix part $\tau_{\text{suf}}=\tau(K+1),\tau(K+2),\dots,\tau(K+S)$, where $\tau(K+S)=\tau(K)$ and $\omega$ denotes the indefinite execution of $\tau_{\text{suf}}$. Among all prefix-suffix schedules that satisfy $\varphi$, we select one that incurs the minimum cost defined as 
\begin{equation}\label{eq:cost}
J(\tau)=\hat{J}(\tau_{\text{pre}}) + \hat{J}(\tau_{\text{suf}}).
\end{equation}
In other words, our method for solving
\begin{align*}
\mathrm{min}~J(\tau) \quad \text{such that} \quad \tau \models \varphi,
\end{align*}
is exhaustive search.
\subsection*{Laplacian Consensus}
The motivation for goal (i) i.e.~activating every link infinitely often is elucidated when sensors are tasked with propagating data they collect according to the link activation schedule. Let $\mathbf{y}(t) \in \mathbb{R}^N$ be a column-vector of states held by each sensor in the network at discrete-time instant $t$. Under a simple Laplacian consensus rule, convergence to a uniform consensus state depends on the topology of the switching network, a phenomenon well-studied. One sufficient condition to reach consensus is the following:
\begin{prop}[\cite{jadbabaie2003}[Theorem 2]
Suppose $\mathcal{G}(t)$ is a switching network and there is an infinite sequence of contiguous, nonempty, bounded intervals $\left[ t_i, t_{i+1} \right)$ starting at $t=0$ such that each union
\[
\bigcup_{t \in \left[ t_i, t_{i+1} \right)} \mathcal{G}(t)
\]
is connected.
Then, $\lim_{t \to \infty} \mathbf{y}(t) = \alpha \mathbf{1}$ where $\alpha$ depends only on $\mathbf{y}(0)$ and the topology of $\mathcal{G}(t)$.
\end{prop}

After seeding the sensor network with initial state $\mathbf{y}(0)$, we use a simple Laplacian consensus rule to quantify the propagation of information across the network under LNC. For a trace $\tau = \tau(0) \tau(1) \cdots$, let $\graph(t)$ denote the subgraph of edges activated at time $t$. Let $\mathcal{V}(t)$ denote the vertices of $\mathcal{G}(t)$, and let $\mathcal{N}_{i}(t)$ denote the neighbors of node $i \in \mathcal{V}$ at time $t$. If the network is to reach a consensus state, it is necessary (but not sufficient) that a trace $\tau$ satisfies the liveness property,
\begin{equation}\label{eq:liveness}
\bigwedge_{i j \in \mathcal{E}} \always \eventually \pi^{i j}.
\end{equation}
A shrewder characterization of Laplacian consensus for switching networks controlled by LTL constraints is fertile ground for future research.

In Section \ref{sec:simulations}, we will implement a simple Laplacian update rule for a sensor network satisfying LNC.
Let $L(t)$ be the (un-normalized) \textit{graph Laplacian} \cite{chung1997} of the switching network $\mathcal{G}(t)$.
For step-size $\epsilon \in \left(0, 1\right)$, the Laplacian update is given by $\vec{y}(t+1) = \vec{y}(t) - \epsilon L(t) \vec{y}(t)$.
Specifically, for node $i$, this is the message-passing scheme,
\begin{equation}
\label{eq:consensus}
y_i(t+1) = 
\begin{cases}
y_i(t)  & \mathcal{N}_i(t) \cap \nodes(t)  = \emptyset \\
(1 -\epsilon) y_i(t) + \epsilon y_j(t) & \mathcal{N}_i(t) \cap \mathcal{V}(t) \neq \emptyset
\end{cases}.
\end{equation}

\subsection*{Summary}
The problem we address in this paper can be summarized as follows.
\begin{prob}\label{prob:centralized}
Compute a schedule $\tau$ in a prefix-suffix form that satisfies the communication task captured in \eqref{eq:global-phi} and minimizes the cost function \eqref{eq:cost}.
\end{prob}
\label{sec:problem}
\section{Design of Network Schedule}
\subsection{Centralized}
\label{sec:feasible}
Existing approaches can be employed to solve Problem \ref{prob:centralized} that rely on a coupling between the transition system, $\vec{TS}$, and what is called a \buchi automaton.
\begin{defn}
A \textit{Non-Deterministic \buchi Automaton} (NBA) is a tuple, $B = \left(S, \Sigma, \delta, S_0, F \right)$ where
\begin{itemize}
\item $S$ is a finite set of states,
\item $\Sigma$ is an input alphabet,
\item $\delta: S \times \Sigma \rightarrow \powerset{S}$, a transition map,
\item $S_0 \subseteq S$ is a set of initial states, and
\item $F \subseteq S$ is a set of accepting states.
\end{itemize}
\end{defn}

The global transition system and NBA are coupled so that $\Sigma = \powerset{\atomic}$.
We say that \textit{$B$ accepts an input word}, $\sigma = \sigma_0 \sigma_1 \sigma_2 \cdots \in \left( \powerset{\atomic}\right)^{\omega}$, if there is at least one (possibly many) sequence of states, $s = s(0) s(1) s(2) \cdots \in S^{\omega}$, such that $s(0) \in S_0$ and $\card{\setof{t: s(t) \in F}} = \infty$.

We say an LTL sentence $\varphi$ is \textit{translated} into an NBA, $B_\varphi$, if a $\sigma \models \varphi$ if and only if $\sigma$ is accepted by $B_\varphi$. It is known that every LTL sentence can be translated into an NBA (for a proof see \cite{baier2008}[p.~278]).
Once we have translated $\varphi$ into $B_{\varphi}$, we construct the \textit{Product \buchi Automaton} (PBA), $\vec{P} = \vec{TS} \bigotimes B_{\varphi}$ (see \cite{baier2008}[p.~200] or \cite{belta2017}[p.~42]).
After quotienting out the underlying digraph of $\vec{P}$ by strongly connected components, we perform a graph search for paths beginning at an initial state $\vec{q}_0$ that visit an accepting state in $Q \times F$ infinitely often.
Projecting this path in $\vec{P}$ onto $\vec{TS}$ yields a trace $\tau \models \varphi$.
The corresponding output word, $\tau_o$, is an infinite sequence of subsets of activated edges solving Problem \ref{prob:centralized}.

\subsection{Hierarchical}
\label{sec:distributed}
The centralized (optimal) solution to Problem \ref{prob:centralized} on $\graph$ is resource-demanding and computationally expensive, unmanageable even for relatively small networks as confirmed by our experiments.
Even if larger networks are manageable with more powerful computers, the centralized algorithm is not scalable: our goal in this section is to solve Problem \ref{prob:centralized} in a more computationally efficient way.

In the sequel, we do not assume sensors are capable of computing temporal logic plans; hence, we will employ additional agents called \textit{commanding agents} or \textit{command nodes} with this capability. Commanding agents form an intermediate layer of schedulers who compute schedules for sensors under their respective jurisdictions, then push theses schedules down to the executing sensor nodes.

We offer a robust approach---in principle, improving scalability---in which (additional) commanding agents (i) solve the local LNC problem for all sensors in their coverage region, and (ii) a global LNC problem is solved on the network formed by these commanding agents (Fig. \ref{fig:cartoon}). As an additional feature to the hierarchical approach, if new sensors are deployed or existing sensors fail, new plans need only be re-synthesized by commanding agents in range of these sensors.

We assume commanding agents have computing capabilities and are able to communicate with all sensor and command agents in range.
We also assume command agents have geospatial awareness. (Recall, we do not assume rank-and-file sensors know their locations or command nodes' locations.)  
Command agents, once deployed, remain at positions, $c_j$, $j \in \setof{1, \dots, K}$. Command agents, we assume, also reside in the (compact) domain $\sensors{Z}$.

Once activated, command agents can detect unique identifiers of rank-and-file sensors within their range $R$ as well as control the activation of links between these sensors.
Additionally, each commanding agent (node) is capable of communicating with peer commanding agents (nodes) whose coverage regions overlap (i.e.~within a range of $2R$), forming their own network $\mathcal{G}_{\mathrm{cmd}} = (\mathcal{V}_{\mathrm{cmd}}, \mathcal{E}_{\mathrm{cmd}})$ where $\mathcal{V}_{\mathrm{cmd}} = \setof{1, \dots, K}$ is the set of commanding nodes and $\mathcal{E}_{\mathrm{cmd}} = \setof{i j \vert~ \| c_i - c_j \| \leq 2 R} \subseteq \mathcal{V}_{\mathrm{cmd}} \times \mathcal{V}_{\mathrm{cmd}}$. Due to our assumption that command nodes have GPS capabilities, it is fair to assume that a given command node $j \in \mathcal{V}_{\mathrm{cmd}}$ is aware of the topology of the connected component of $\mathcal{G}_\mathrm{cmd}$ containing $j$. (Command nodes can transmit their unique identifiers and GPS coordinates to command nodes in range, this data and $R$ is enough to discern the topology of a connected component of $\mathcal{G}_\mathrm{cmd}$.)

\begin{figure}
\begin{center}
\includegraphics[width=\linewidth]{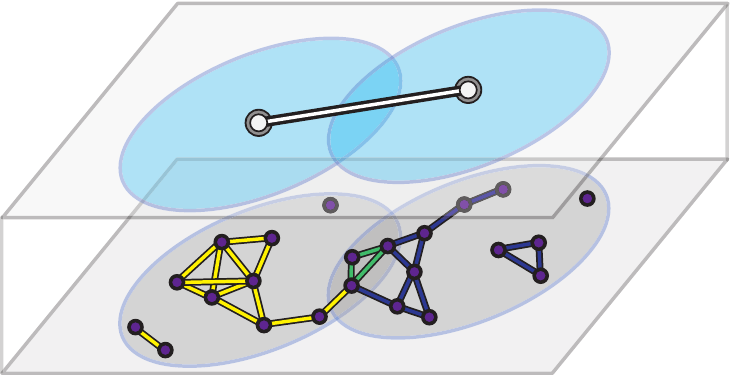}
\end{center}
\caption{Sensor network with two command nodes (white). Schedules $\tau_1$, $\tau_2$ for $\mathcal{G}_1$ (yellow/green edges) and $\mathcal{G}_2$ (blue/green edges) synthesized by command nodes $1$ and $2$ respectively. In this simple case, command nodes $1$ and $2$ alternate executing $\tau_1$ and $\tau_2$.}
\label{fig:cartoon}
\end{figure}

Local plans are synthesized as follows:
command nodes first collect sensor nodes within their range, $R$.
For each $j \in \mathcal{V}_{\mathrm{cmd}}$, let $\mathcal{G}_j = (\mathcal{V}_j, \mathcal{E}_j)$ be the subgraph of $\mathcal{G}$ generated by sensor nodes,
\[
\setof{i \in \mathcal{V} \vert~ \|x_i - c_j\| \leq R}.
\]
Then, commanding nodes compute a local specification $\varphi_j$ given by \eqref{eq:global-phi}--\eqref{eq:local-phi}, restricted to $\mathcal{G}_j$.

Each command node $j$ synthesizes an optimal plan $\tau_j \models \varphi_j$ (Section \ref{sec:feasible}); explicitly, a feasible prefix/suffix $\tau_{j, \mathrm{pre}}, \tau_{j, \mathrm{suf}}$ that minimizes the cost-to-go $J(\cdot)$. Command nodes, once activated, execute this plan.
Care must be taken to avoid activating adjacent command nodes simultaneously, as their coverage regions overlap, possibly leading to link-activation collisions. We employ LTL planning (Section \ref{sec:feasible}) to synthesize a (feasible) trace of command node activations; such a trace $\rho = \rho_{\mathrm{pre}} \left[\rho_{\mathrm{suf}}\right]^{\omega}$ is said to be feasible if $\rho \models \psi$ where
\begin{equation}
\label{eq:command-plan}
\psi = \bigconj_{j \in \mathcal{V}_{\mathrm{cmd}}} \left( \always \left( \pi^{j} \Rightarrow \bigconj_{k \in \mathcal{N}_{j}} \neg \pi^{k} \right) \conj \always \eventually \pi^{j} \right).
\end{equation}
$\psi$ specifies that every command node is activated infinitely often, and that whenever a command node is activated, every command node in its $1$-hop neighborhood is deactivated. Each command node computes a global plan of command node activations $\rho$. Any discrepancy in $\rho$ across $\mathcal{G}_{\mathrm{cmd}}$ is remedied by a simple consensus algorithm (e.g.~\cite{lamport1998}). Full details are described in Algorithm \ref{alg:sensor}.

\begin{algorithm}[h]
\caption{Hierarchical LNC (HLNC)}
\label{alg:sensor}
\begin{algorithmic}[1]
\renewcommand{\algorithmicrequire}{\textbf{Input:}}
\renewcommand{\algorithmicensure}{\textbf{Output:}}
\REQUIRE $\mathcal{G}$ (sensor net.), $x_i$ (locations of sensors), $c_j, j \in \{1,2, \cdots, K\}$ (locations of command agents), $J$ (cost-to-go)
\ENSURE $\tau \in \left(\powerset{\mathcal{AP}}\right)^{\omega}$ (global schedule of link activations)
\STATE command graph $\mathcal{G}_{\mathrm{cmd}}$ constructed by rule: $jk \in \edges_{\mathrm{cmd}}$ if $B_R(c_j) \cap B_R(c_k) \neq \emptyset$; command nodes in every connected component of $\mathcal{G}_{\mathrm{cmd}}$ aware of topology
\FOR{command nodes $j =1$ \TO $K$}
\STATE command node activation plan $\rho$ synthesized by $j$ according to $\psi$
\STATE verify consensus $\rho$ reached across $\mathcal{G}_{\mathrm{cmd}}$
\STATE $\mathcal{G}_j \leftarrow$ subgraph induced by $\setof{i \in \mathcal{V}~\vert~ \|c_j - x_i\| \leq R}$
\STATE $\varphi_j$ given by $\varphi_{\vert \mathcal{G}_j}$ translated to \buchi by $j$
\STATE command node $j$ computes all $\tau_j \models \varphi_j$
\STATE choose $\tau_j$ minimizing $J(\cdot)$
\ENDFOR
\RETURN $\tau(t) := \bigcup_{j \in \rho(t)} \tau_j(t)$
\end{algorithmic}
\end{algorithm}

\paragraph*{Remark}
Our distributed algorithm avoids the trivial feasible solution---activate every edge in the network one-at-a-time---because command nodes control the activations of multiple links in their domain, and multiple command nodes are activated simultaneously.
We call such trivial trace $\tau_{\mathrm{seq}}$. It is true \textit{a priori} $\tau_{\mathrm{seq}} \models \varphi$. However, $ \tau_{\mathrm{seq}}$ is highly suboptimal, e.g.~for Jaccard distance (the cost of every transition is maximal i.e.~$1$). While we will prove below that a feasible trace is returned by Algorithm \ref{alg:sensor}, such a trace is nevertheless suboptimal
because minimizing the cost of each summand in
\[
J(\tau) = \sum_{j \in \mathcal{V}_{\mathrm{cmd}}} J(\tau_j)
\]
is not equivalent to minimizing the entire sum.

\subsection{Analysis}
In order that Algorithm \ref{alg:sensor} returns a trace satisfying $\varphi$ \eqref{eq:global-phi}, it is integral that every link of the sensor network is detected by at least one command node. We supply a sufficient condition for this property to hold.
\begin{prop}
\label{thm:edge-contained}
Suppose there exists $\epsilon \in \left(r, R \right)$ such that
\[
\bigcup_{j \in \mathcal{V}_{\mathrm{cmd}}} B_{R-\epsilon}(c_j) \supseteq \sensors{X}
\]
is a finite cover of $\sensors{X}$ by command nodes at locations $c_j, j \in \setof{1, \dots, K}$.
Then, every edge $ii' \in \mathcal{E}$ is contained in (at least one) command subgraph $\mathcal{G}_j$ for some $j \in \mathcal{V}_{\mathrm{cmd}}$.
\end{prop}
\begin{proof}
A finite cover exists by compactness of $\mathbf{Z}$. Let $x_{i}, x_{i'} \in \sensors{X}$, and $i i' \in \mathcal{E}$.  By assumption, there is some $j$ such that $B_{R-\epsilon}(c_j)$ contains $x_{i}$ and some $j'$ such that $B_{R-\epsilon}(c_{j'})$ contains $x_{i'}$. By assumption, $\lVert x_i - x_{i'} \rVert \leq r$. Likewise, $\lVert  c_{j} - x_{i} \rVert \leq R - \epsilon$ so that 
\begin{align*}
\lVert c_j - x_{i'} \rVert  \leq \lVert c_j - x_{i} \rVert + \lVert x_{i} - x_{i'} \rVert \\
\leq R - \epsilon + r  < R
\end{align*}
by the triangle inequality.
Hence, $B_R(c_j)$ contains $x_i$ and $x_{i'}$. Thus, $\mathcal{G}_j$ contains the link $ii' \in \mathcal{E}$.
\end{proof}
As a convenience, we may assume without loss of generality (under the conditions above) that the commanding network is connected; Algorithm \ref{alg:sensor} can be run in parallel on each connected component of $\mathcal{G}_{\mathrm{cmd}}$ without modification.
\begin{cor}
Connected components of $\mathcal{G}$ are controlled only by command nodes belonging to the same connected component of $\mathcal{G}_{\mathrm{cmd}}$.
\end{cor}
\begin{proof}
Suppose, on the contrary, there are nodes $i, i' \in \mathcal{V}$ such that $x_i$ and $x_{i'}$ are contained in different connected components of $\mathcal{U}:=\bigcup_{j \in \mathcal{V}_{\mathrm{cmd}}} B_R(c_j)$, but $\lVert x_i - x_{i'} \rVert \leq r$. Then, by Proposition \ref{thm:edge-contained}, there is a single $j$ such that $B_R(c_j) \supseteq \setof{x_i, x_{i'}}$. In particular, this means $x_i$ and $x_{i'}$ lie in the same component of $\mathcal{U}$, hence, $\mathcal{G}_{\mathrm{cmd}}$.
\end{proof}
The HLNC algorithm guarantees a feasible suboptimal solution.
\begin{thm}
\label{thm:completeness}
Suppose the hypothesis of Proposition \ref{thm:edge-contained} is satisfied, and suppose without loss of generality that $\mathcal{G}_{\mathrm{cmd}}$ is connected. Then, Algorithm \ref{alg:sensor} returns a feasible trace $\tau \models \varphi$ if such a trace exists.
\end{thm}
We sketch the proof.
\begin{proof} We show the output $\tau$ satisfies parts (a) and (b)---local non-interference and liveness---of $\varphi_{i j}$ \eqref{eq:local-phi} for every link $i j \in \mathcal{E}$.

By construction, whenever a command node $j$ is activated and initiates a local plan $\tau_j$ satisfying $\varphi_j$, neighboring command nodes are deactivated. This implies that $\tau_j$ extends properly to a global plan satisfying $\varphi$ because $R$-balls centered around simultaneously-activated command nodes intersect trivially.

Command nodes are activated infinitely often according to a schedule $\rho \models \psi$. By Proposition \ref{thm:edge-contained}, every link of the sensor network is contained in a command subgraph $\mathcal{G}_j$ in which every edge is activated infinitely often. Hence, every edge of the entire network is activated infinitely often.
\end{proof}
A final result contrasts the computational complexity of the centralized and hierarchical algorithms.

\begin{thm}
Let $e_{\mathrm{max}} = \mathrm{max}_{j \in \mathcal{V}_{\mathrm{cmd}}} | \mathcal{E}_j |$. Assuming the hypothesis of Proposition \ref{thm:edge-contained}, the worst-case time complexity of Algorithm \ref{alg:sensor} is
\[
\mathcal{O} \left( K 2^{\mathcal{O}(e_\mathrm{max}) + \mathcal{O}(K)} \right).
\]
The worst-case time complexity of the centralized algorithm is
\[
\mathcal{O} \left( 2^{\mathcal{O}(|\mathcal{E}|)} \right).
\]
\end{thm}
Again, we sketch a proof.
\begin{proof}
The complexity of solving the LTL model-checking problem (equivalent to the optimal temporal logic planning problem) for a proposition $\varphi$ and a transition system with states $Q$ is $\mathcal{O}\left( |Q| 2^{|\varphi|} \right)$ where $|\varphi|$ is the length of the proposition (roughly, the number of symbols) \cite{belta2017}.

For the centralized algorithm, we check that $|\varphi| = \mathcal{O}(|\mathcal{E}|)$ and $|Q| = 2^{|\mathcal{E}|}$ yielding $\mathcal{O} \left( 2^{\mathcal{O}(|\mathcal{E}|)} \right)$ complexity. For the hierarchical algorithm, we first consider the complexity of satisfying $\psi$ \eqref{eq:command-plan}.
This is $\mathcal{O}\left( K 2^{|\psi|} \right) = \mathcal{O}\left( K 2^{\mathcal{O}(K)} \right)$.
We multiply this complexity by the worst-case execution of the centralized algorithm to find a feasible plan satisfying local $\varphi_j$ for each command node $j$ to yield $\mathcal{O} \left( K 2^{\mathcal{O}(e_{\max}) + \mathcal{O}(K)} \right)$.
\end{proof}
As expected, the Algorithm \ref{alg:sensor} boasts a significant improvement to the centralized algorithm when both $K$ and $e_{\mathrm{max}}$ are not too large as compared to $|\mathcal{E}|$.
\label{sec:solution}
\section{Simulations}
\label{sec:simulations}
\subsection*{Software}
We use the package \texttt{P\_MAS\_TG} \cite{guo2015} for optimal temporal logic planning.
The centralized component of our code calculates a $\varphi$ from an input graph $\graph$, then proceeds to translate $\varphi$ and find an optimal plan.
\texttt{P\_MAS\_TG} utilizes in the background the package \texttt{ltl2ba} \cite{gastin2001} that translates a given LTL sentence into a \buchi automaton.
The decentralized component of our code implements Algorithm \ref{alg:sensor} and stores $\mathcal{G}_j$, $\tau_{j}$ as \texttt{networkx} \cite{hagberg2008} attributes of command nodes.
\subsection*{Finding an optimal trace}
A sampling algorithm, such as the one presented in \cite{kantaros2018}, can find a feasible $\tau$ than approximately minimizes $J(\cdot)$.
We do not implement this algorithm here as it is beyond the scope of this paper.
Instead, we use the out-of-the-box path distance minimizing feature of \texttt{P\_MAS\_TG}.
The package \texttt{P\_MAS\_TG} minimizes traveling distance between given region coordinates (e.g. locations in a warehouse for a robot to visit).
\subsection*{Command nodes}
We select positions for $K=10$ commanding agents. To ensure these command nodes cover every sensor in the network, we chose cluster centers $c_j$ of the positions of rank-and-file sensors $x_i$ via k-means, which selects a partition $\{S_1, \dots, S_K \}$ of the sensors $\sensors{X}$ that minimizes the distance of the means (i.e.~command nodes)
\[ c_j = \frac{1}{|S_j|} \sum_{x \in S_j} x \]
 to points (i.e.~sensor nodes) in each cluster.
Clusters and their respective means are computed by solving the optimization problem
\begin{equation}\label{eq:kmeans}
\min_{S_1, \dots, S_K} \sum_{j=1}^K \sum_{x \in S_j} \lVert x - c_j \rVert
\end{equation}
by standard methods \cite{lloyd1982}.
\subsection*{Data}
We study a real-world sensor network $N = 54$ located inside the Berkeley National Laboratory \cite{madden2004}. We assume sensors communicate at a radius $r=6$ and commanding agents can control link connectivity within $R = 8$. The resulting graphs $\mathcal{G}$, $\mathcal{G}_{\mathrm{cmd}}$ have $88$ and $13$ edges, respectively.
(Fig.~\ref{fig:sensor-net}).
\begin{figure}[h]
\begin{center}
\includegraphics[width=\linewidth]{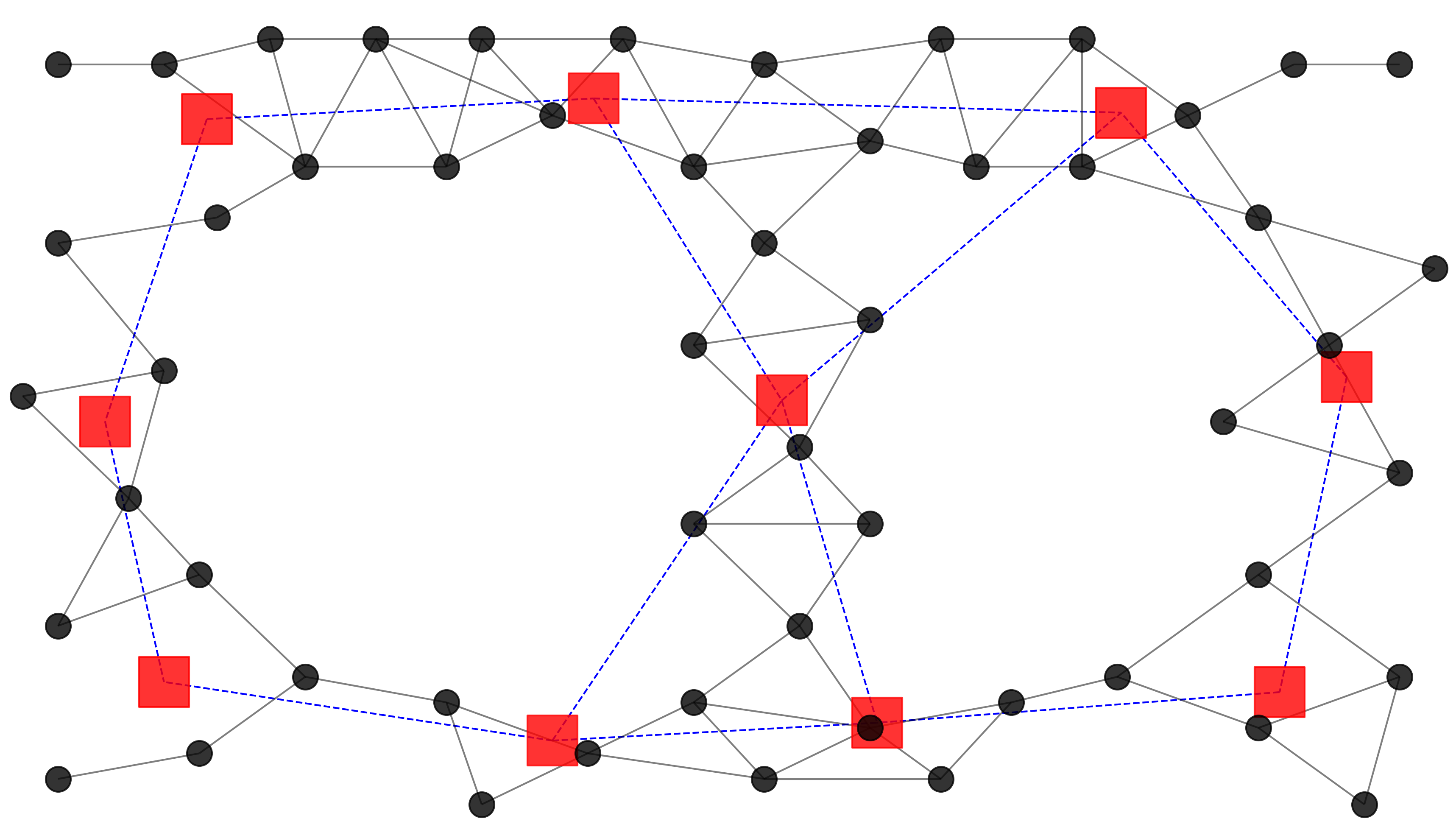}
\end{center}
\caption{Sensor network consisting of $54$ sensors in Berkeley National Lab with rank-and-file sensors (black) and command agents (red).}
\label{fig:sensor-net}
\end{figure}

\subsection*{Discussion of results}
We compare intermittent consensus ($\epsilon = 0.5$) of the computed global trace $\tau$ (Algorithm \ref{alg:sensor}) with intermittent consensus of the sequential trace $\tau_{\mathrm{seq}}$. Our simulation demonstrates that HLNC rapidly achieves localized consensus (Fig.\ \ref{fig:consensus}), locally converging faster than $\tau_{\mathrm{seq}}$; however, $\tau$ (globally) falls behind $\tau_{\mathrm{seq}}$ in the long run and may converge to more than one point of consensus (two, in our case). Perhaps the trade-off here is that the efficiency (link activations per unit time) of $\tau$ ($\approx 4-8 \%$) is far greater than the efficiency ($\approx 1 \%$) of $\tau_{\mathrm{seq}}$. 

\begin{figure}[h]
\begin{center}
\includegraphics[width=\linewidth]{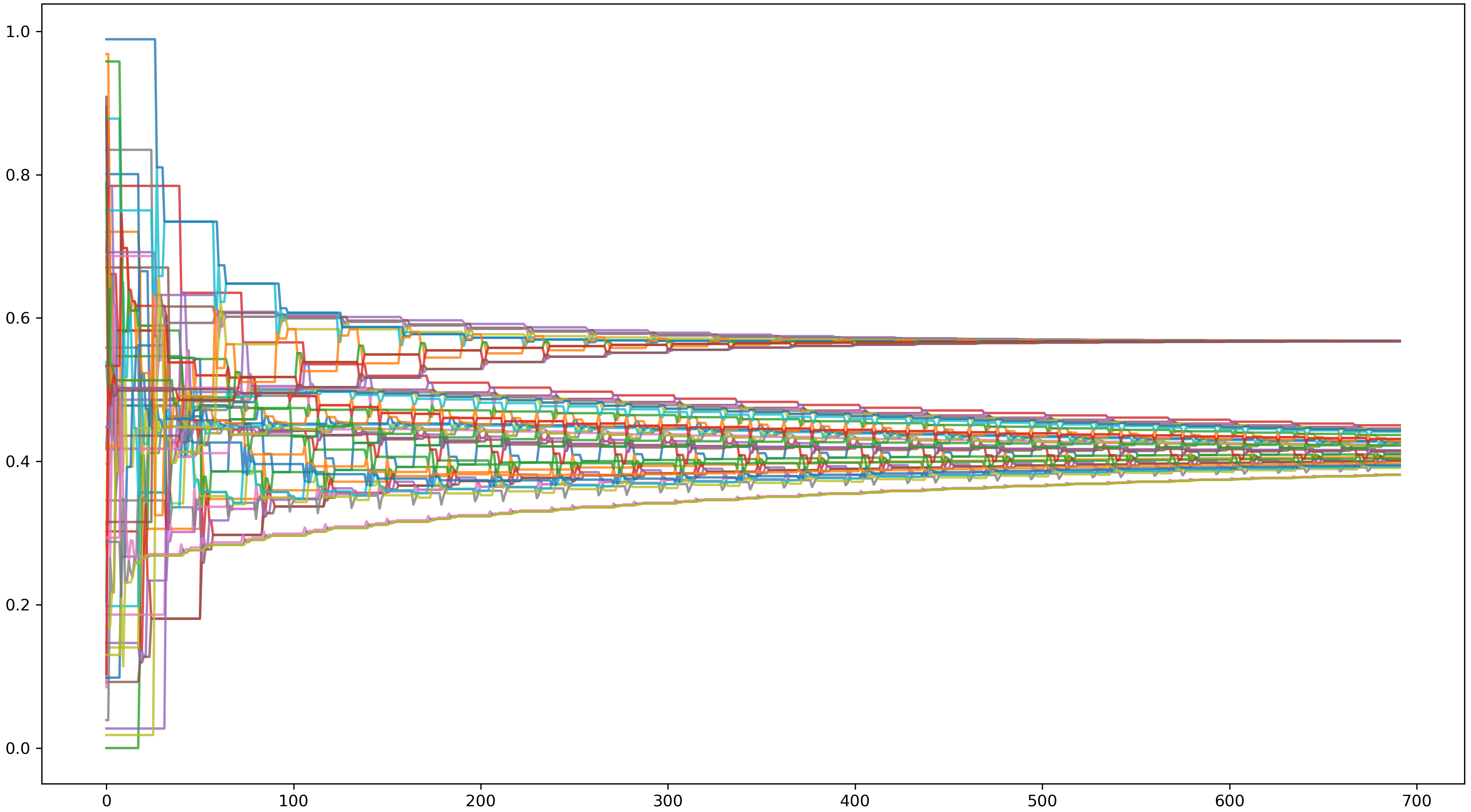}
\caption{Intermittent consensus \eqref{eq:consensus} of sensors following $\tau \models \varphi$ computed with Algorithm \ref{alg:sensor}.}
\label{fig:consensus}
\end{center}
\end{figure}

\section{Discussion}
\label{sec:discussion}
\subsection*{Semantic Propagation}
It is perhaps a compelling question whether there is a general framework for networked temporal logic planning for a specification of the form
\[
\varphi = \bigwedge_{i \in \nodes} \varphi_i
\]
studied both in this work and in \cite{kantaros2016}.

A truly decentralized approach would be to compute LTL plans on individual nodes of the network, rather than on command nodes.
Recent work develops a new theory of Laplacian flow (i.e. a consensus algorithm) for lattice-valued data (e.g.~boolean data, LTL specifications and more) over a network \cite{ghrist2020}.
Avoiding the mathematical methods of \cite{ghrist2020}, the LTL Laplacian can be interpreted as the map
\[
(L \varphi)_i = \bigconj_{j \in \mathcal{N}_i \cup i} \varphi_j
\]
taking the conjunction of $\varphi_i$ with all formulae over the neighborhood of $i$.
Semantic propagation is the result of iterating this Laplacian, eventually leading to locally constant (i.e. consensus over connected components) semantics over the entire network.

\subsection*{Non-geometric networks}
Suppose $\graph$ is not a sensor network i.e. nodes are not given as coordinates and edges are not defined by their proximity. Then, there is a modification of Algorithm \ref{alg:sensor} to solve Problem \ref{prob:centralized} in this case.
Define the \textit{$k$-hop neighborhood of a node}, $\nbhd{k}{i}$.
Fix $k>0$ (analogous to $R$). Choose command nodes $\mathcal{V}_{\mathrm{cmd}} \subseteq \mathcal{V}$. For each $j$ in  $\mathcal{V}_{\mathrm{cmd}}$, let $\mathcal{G}_j$ be the subgraph generated by $\mathcal{N}_k(j)$. Then, $\mathcal{E}_{\mathrm{cmd}} = \setof{ij \in \mathcal{V}_{\mathrm{cmd}} \times \mathcal{V}_{\mathrm{cmd}}~\vert~\mathcal{N}_k(i) \cap \mathcal{N}_k(j) \neq \emptyset}$. The rest of the algorithm holds as above.

\subsection*{Fairness}
As astutely noted by a referee, the liveness specification \eqref{eq:liveness} does not ensure that some nodes are not granted special privilege to connect with other nodes more frequently. A possible solution is to add a (local) fairness specificaiton,
\[
\chi_{i j} = \always \left(\pi^{ij} \Rightarrow \always \left( \neg \pi^{i j} \mathcal{U} \left( \bigvee_{i'j' \in \mathcal{N}_{ij} \setminus \{ij\}} \pi^{i'j'} \right) \right) \right).
\]
\section{Acknowledgments}
 The authors would like to thank the reviewers for their helpful feedback. The authors [HR, RG] were supported by Office of the Assistant Secretary of Defense Research \&  Engineering through a Vannevar Bush Faculty Fellowship, ONR N00014-16-1-2010. The authors [YK, GP] were supported by ARO DCIST and AFOSR Assured Autonomy.


\begin{thebibliography}{99}
\footnotesize
\bibitem{agiwal2016}
Agiwal, M., Roy, A., \& Saxena, N. (2016). Next generation 5G wireless networks: a comprehensive survey. \textit{IEEE Communications Surveys \& Tutorials}, 18(3), 1617 -- 1655.
\bibitem{akyildiz2002}
Akyildiz, I. F., Su, W., Sankarasubramaniam, Y., \& Cayirci, E. (2002). Wireless sensor networks: a survey. \textit{Computer Networks}, 38(4), 393 -- 422.
\bibitem{baier2008}
Baier, C., \& Katoen, J. P. (2008). {\it Principles of model checking}. MIT Press, Cambridge.
\bibitem{belta2017}
Belta, C., Yordanov, B., \& Gol, E. A. (2017). {\it Formal methods for discrete-time dynamical systems} (Vol.~89). Springer, Cham.
\bibitem{chatzipanagiotis2016}
Chatzipanagiotis, N., \& Zavlanos, M. M. (2016). Distributed scheduling of network connectivity using mobile access point robots. \textit{IEEE Transactions on Robotics}, 32(6), 1333 -- 1346.
\bibitem{chung1997}
F. R. Chung, F.~R.~ \& Graham, F.~C. (1997). {\it Spectral Graph Theory}. American Math-
ematical Society, Providence.
\bibitem{fainekos2005hybrid}
Fainekos, G.~E., Kress-Gazit, H., \& Pappas, G.~J.~ (2005). Hybrid controllers for path planning: a temporal logic approach. \textit{In Proceedings of the 44th IEEE Conference on Decision and Control}, 4885 -- 4890.
\bibitem{gastin2001}
Gastin, P., \& Oddoux, D. (2001). Fast LTL to Büchi automata translation. In {\it International Conference on Computer Aided Verification}, 55 -- 56.
\bibitem{ghrist2020}
Ghrist, R., \& Riess, H. (2020). Cellular sheaves of lattices and the Tarski Laplacian. {\it arXiv}:2007.04099.
\bibitem{guo2015}
Guo, M., \& Dimarogonas, D. V. (2015). Multi-agent plan reconfiguration under local LTL specifications. {\it The International Journal of Robotics Research}, 34(2), 218 -– 235.
\bibitem{hagberg2008}
Hagberg, A., Swart, P., \& S Chult, D. (2008). Exploring network structure, dynamics, and function using NetworkX (No. LA-UR-08-05495; LA-UR-08-5495). \textit{Los Alamos National Lab.(LANL)}, Los Alamos, New Mexico.
\bibitem{hollinger2010multi}
Hollinger, G., \& Singh, S. (2010). Multi-robot coordination with periodic connectivity. In \textit{2010 IEEE International Conference on Robotics and Automation}, 4457 -- 4462.
\bibitem{holzmann2004spin}
Holzmann, G.~J.~ (2004). \textit{The SPIN model checker: primer and reference manual} (Vol. 1003). Addison-Wesley, Reading, Pensylvania.
\bibitem{jadbabaie2003}
Jadbabaie, A., Lin, J., \& Morse, A. S. (2003). Coordination of groups of mobile autonomous agents using nearest neighbor rules. \textit{IEEE Transactions on automatic control}, 48(6), 988 -- 1001.
\bibitem{ji2007distributed}
Ji, M., \& Egerstedt, M. (2007). Distributed coordination control of multiagent systems while preserving connectedness. \textit{IEEE Transactions on Robotics}, 23(4), 693 -- 703.
\bibitem{kantaros2016}
Kantaros, Y., \& Zavlanos, M. (2016). Distributed intermittent connectivity control of mobile robot networks. {\it IEEE Transactions on Automatic Control}, 62(7), 3109 -- 3121.
\bibitem{kantaros2018}
Kantaros, Y., \& Zavlanos, M. M. (2018). Sampling-based optimal control synthesis for multirobot systems under global temporal tasks. \textit{IEEE Transactions on Automatic Control}, 64(5), 1916 -- 1931.
\bibitem{kantaros2016distributed}
Kantaros, Y., \& Zavlanos, M.~M. (2016). Distributed communication-aware coverage control by mobile sensor networks. \textit{Automatica}, 63, 209 -- 220.
\bibitem{kloetzer2009automatic}
Kloetzer, M., \& Belta, C. (2009). Automatic deployment of distributed teams of robots from temporal logic motion specifications. \textit{IEEE Transactions on Robotics}, 26(1), 48 -- 61.
\bibitem{lamport1998}
Lamport, L. (1998). The part-time Parliment. \textit{ACM Transactions on Computer Systems}, 16(2), 133 -- 169. 
\bibitem{leahy2017}
Leahy, K.J.~, Aksaray, D., \& Belta, C. (2017). Informative path planning under temporal logic constraints with performance guarantees. {\it IEEE American Control Conference}, 1859 -- 1865.
\bibitem{lloyd1982}
Lloyd, S. P. (1982). Least Squares Quantization in PCM. \textit{IEEE Transactions on Information Theory}, 28(2), 129 -– 137.
\bibitem{li2018}
Li, B., Fei, Z., \& Zhang, Y. (2018). UAV communications for 5G and beyond: recent advances and future trends. \textit{IEEE Internet of Things Journal}, 6(2), 2241 -- 2263.
\bibitem{onella2007}
Onnela, J. P., Saramäki, J., Hyvönen, J., Szabó, G.~et al.~(2007). Structure and tie strengths in mobile communication networks. \textit{Proceedings of the National Academy of Sciences}, 104(18), 7332 -- 7336.
\bibitem{madden2004}
Madden, S. (2004). Intel lab data. Retrieved August 25, 2020, from http://db.csail.mit.edu/labdata/labdata.html.
\bibitem{scikit-learn}
Pedregosa, F., Varoquaux, G., Gramfort et al.~(2011). Scikit-learn: machine learning in Python. \textit{Journal of Machine Learning Research}, 12, 2825 –- 2830.
\bibitem{sabattini2013decentralized}
Sabattini, L., Chopra, N., \& Secchi, C. (2013). Decentralized connectivity maintenance for cooperative control of mobile robotic systems. \textit{The International Journal of Robotics Research}, 32(12), 1411 -- 1423.
\bibitem{smith2011optimal}
Smith, S.~L., Tůmová, J., Belta, C., \& Rus, D. (2011). Optimal path planning for surveillance with temporal-logic constraints. \textit{The International Journal of Robotics Research}, 30(14), 1695 -- 1708.
\bibitem{stephan2017concurrent}
Stephan, J., Fink, J., Kumar, V., \& Ribeiro, A. (2017). Concurrent control of mobility and communication in multirobot systems. \textit{IEEE Transactions on Robotics}, 33(5), 1248 -- 1254.
\bibitem{yan2012robotic}
Yan, Y., \& Mostofi, Y. (2012). Robotic router formation in realistic communication environments. \textit{IEEE Transactions on Robotics}, 28(4), 810 -- 827.
\bibitem{zavlanos2011graph}
Zavlanos, M.~M., Egerstedt, M. B., \& Pappas, G. J. (2011). Graph-theoretic connectivity control of mobile robot networks. \textit{Proceedings of the IEEE}, 99(9), 1525 -- 1540.
\bibitem{zavlanos2007potential}
Zavlanos, M.~M., \& Pappas, G. J. (2007). Potential fields for maintaining connectivity of mobile networks. \textit{IEEE Transactions on robotics}, 23(4), 812 -- 816.
\bibitem{zavlanos2008distributed}
Zavlanos, M.~M., \& Pappas, G. J. (2008). Distributed connectivity control of mobile networks. \textit{IEEE Transactions on Robotics}, 24(6), 1416 -- 1428.
\bibitem{zavlanos2013network}
Zavlanos, M.~M., Ribeiro, A., \& Pappas, G.~J. (2012). Network integrity in mobile robotic networks. \textit{IEEE Transactions on Automatic Control}, 58(1), 3 -- 18.
\bibitem{zavlanos2008}
Zavlanos, M.~M., Tahbaz-Salehi, A., Jadbabaie, A., \& Pappas, G. J.~(2008). Distributed topology control of dynamic networks. In \textit{2008 American Control Conference}, 266 -- 2665.
\bibitem{zhu2009}
Zhu, X., Shen, L., \& Yum, T.~S.~P.~(2009). Hausdorff clustering and minimum energy routing for wireless sensor networks. {\it IEEE Transactions on Vehicular Technology}, 58(2), 990 -- 997.
\end{thebibliography}
\end{document}